%% file: root.tex
\documentclass[journal,twoside,web]{ieeecolor}
\usepackage{generic}
\usepackage{cite}
\usepackage{amsmath,amssymb,amsfonts}
\usepackage{gensymb}
\usepackage{siunitx}
\usepackage{graphicx}
\usepackage{booktabs}
\usepackage{textcomp}
\usepackage{mathtools}
\usepackage{graphicx}
\usepackage[caption=false]{subfig}
\usepackage{xcolor}
\usepackage{hyperref}
\usepackage[nolist,nohyperlinks]{acronym}
\def\BibTeX{{\rm B\kern-.05em{\sc i\kern-.025em b}\kern-.08em
    T\kern-.1667em\lower.7ex\hbox{E}\kern-.125emX}}
    
\newcommand{\sgnote}[1]%
{\textcolor{green}{\textbf{Note: #1}}}
\newcommand{\shnote}[1]%
{\textcolor{blue}{\textbf{Note: #1}}}
\newcommand{\llnote}[1]%
{\textcolor{blue}{\textbf{LL: #1}}}

\newtheorem{lemma}{Lemma}
\newtheorem{theorem}{Theorem}

\newtheorem{assumption}{Assumption}
\newtheorem{remark}{Remark}
\newtheorem{algorithm}{Algorithm}
\newtheorem{proposition}{Proposition}

\usepackage{siunitx} 

\begin{document}

\title{\LARGE \bf
Classical Risk-Averse Control for a Finite-Horizon Borel Model
}
\author{Margaret P. Chapman$^\dag$ and Kevin M. Smith$^\ddag$
\thanks{$^\dag$M.P.C. is with the Edward S. Rogers Sr. Department of Electrical and Computer Engineering, University of Toronto, 10 King's College Rd, Toronto, ON M5S 3G8, Canada (email: mchapman@ece.utoronto.ca).}
\thanks{$^\ddag$K.M.S. is with the Department of Civil and Environmental Engineering, Tufts University, 200 College Ave, Medford, MA 02155, United States (email: kevin.smith@tufts.edu).}
\thanks{M.P.C. acknowledges support from the University of Toronto.}
}

\maketitle
\pagestyle{empty}
\thispagestyle{empty}
\begin{abstract}We study a risk-averse optimal control problem for a finite-horizon Borel model, where a cumulative cost is assessed via exponential utility. The setting permits non-linear dynamics, non-quadratic costs, and continuous state and control spaces but is less general than the problem of optimizing an expected utility. Our contribution is to show the existence of an optimal risk-averse controller without using state space augmentation and therefore offer a simpler solution method from first principles compared to what is currently available in the literature.\end{abstract}

\input{1_introduction}
\input{2_problemstatement}

\input{3_algorithm}
%
\input{Appendix}
\input{6_conclusions}

\section*{Appendix}
We state and prove Lemma \ref{phiprimeislsc} below.
\begin{lemma}[$\phi_t'$ is lsc and finite]\label{phiprimeislsc}
Recall that $\phi_t : S \times A \times D \rightarrow (0,+\infty)$ is lsc and bounded \eqref{myphi}, and $p_t(\mathrm{d}w_t|x_t,u_t)$ is a continuous stochastic kernel on $D$ given $S \times A$. It holds that $\phi_t' : S \times A \rightarrow (0,+\infty)$ \eqref{myphitprime} is lsc.
\end{lemma}
\begin{proof}
Showing that $\phi_t'$ is lsc is a special case of \cite[Prop. 7.31]{bertsekas2004stochastic}, which we call Corollary 1: \emph{Let $\mathcal{X}$ and $\mathcal{Y}$ be separable metrizable spaces, and let $q(\mathrm{d}y|x)$ be a continuous stochastic kernel on $\mathcal{Y}$ given $\mathcal{X}$. Suppose that $g : \mathcal{X} \times \mathcal{Y} \rightarrow \mathbb{R}$ is lsc, and there are scalars $\underline{c}$ and $\overline{c}$ such that $\underline{c} \leq g(x,y) \leq \overline{c}$ for all $(x,y) \in \mathcal{X} \times \mathcal{Y}$.
Then, the function $\lambda : \mathcal{X} \rightarrow \mathbb{R}$ defined by $\lambda(x) := \int_{\mathcal{Y}} g(x,y) \; q(\mathrm{d}y|x)$ is lsc.} 

To apply Corollary 1, choose $\mathcal{X} = S \times A$, $\mathcal{Y} = D$, $g = \phi_t$, $\lambda = \phi_t'$, $\underline{c} = e^{\frac{-\theta}{2}\underline{b}}$, $\overline{c} = e^{\frac{-\theta}{2}\overline{b}}$, where $\underline{b}$ is a lower bound and $\overline{b}$ is an upper bound for $V_{t+1}^\theta$, and $q(\mathrm{d}y|x) = p_t(\mathrm{d}w|x,u)$. 

To prove Corollary 1, one uses the fact that $g : \mathcal{X} \times \mathcal{Y} \rightarrow \mathbb{R}$ being lsc and bounded ($\underline{c} \leq g \leq \overline{c}$), where $\mathcal{X} \times \mathcal{Y}$ is a metrizable space, implies that there is a sequence of continuous functions $g_m : \mathcal{X} \times \mathcal{Y} \rightarrow \mathbb{R}$ such that $\underline{c} \leq g_m \leq g_{m+1} \leq g \leq \overline{c}$ for all $m \in \mathbb{N}$, and $\{g_m\}_{m=1}^\infty$ converges to $g$ pointwise. The proof of this fact uses techniques from \cite[Lemmas 7.7 \& 7.14]{bertsekas2004stochastic} and \cite[Theorem A6.6]{ash1972}. We outline some key steps below for clarity.

One may choose $g_m(z) := \inf\{g(s) + m \rho(z,s) : s = (s_1,s_2), s_1 \in \mathcal{X}, s_2 \in \mathcal{Y} \}$, where $z = (z_1,z_2)$, $z_1 \in \mathcal{X}$, $z_2 \in \mathcal{Y}$, $m \in \mathbb{N}$, and $\rho$ is an appropriate metric on $\mathcal{X} \times \mathcal{Y}$. Now, let $z = (z_1,z_2) \in \mathcal{X} \times \mathcal{Y}$ and $\epsilon > 0$ be given. For each $m \in \mathbb{N}$, $g_m(z)$ is finite, and hence, there is a point $z_m = (z_{1m},z_{2m}) \in \mathcal{X} \times \mathcal{Y}$ such that 
\begin{equation}
    g(z_m) + m \rho(z, z_m) \leq g_m(z) + \epsilon.
\end{equation}
Since $\underline{c} \leq g(z_m)$ and $g_m(z) \leq g(z)$ for all $m \in \mathbb{N}$, we have
\begin{equation}\label{my24}
    \underline{c} + m \rho(z, z_m) \leq g_m(z) + \epsilon \leq g(z) + \epsilon \;\;\; \forall m \in \mathbb{N}.
\end{equation}
Since $\underline{c}$ is finite, $m$ is finite and positive, $\rho$ is bounded below by zero, and from \eqref{my24}, it follows that
\begin{equation}\label{my34}
    0 \leq \rho(z, z_m) \leq \frac{g(z) + \epsilon - \underline{c}}{m} \;\;\;\;\forall m \in \mathbb{N}.
\end{equation}
The inequality \eqref{my34} and $g(z)$ being finite imply that
\begin{equation}
    \liminf_{m \rightarrow \infty} \rho(z, z_m) = \limsup_{m \rightarrow \infty} \rho(z, z_m) = 0,
\end{equation}
which shows that the limit of $\{\rho(z, z_m)\}_{m=1}^\infty$ exists and equals zero. Then, $g$ being lsc and $\rho(z, z_m) \rightarrow 0$ implies that $g(z) \leq \underset{m \rightarrow \infty}{\liminf} \; g(z_m)$.

Moreover, the proof of Corollary 1 uses the Extended Monotone Convergence Theorem \cite[p. 47]{ash1972}, which applies in particular because $g_m \geq \underline{c}$ for all $ m \in \mathbb{N}$.
\end{proof}

\section*{Acknowledgements}
The authors thank Riccardo Bonalli, Marco Pavone, and Dimitri Bertsekas for discussions. The authors are also grateful for advice provided by the Associate Editor and three reviewers, whose comments improved the presentation of this work.

\bibliographystyle{IEEEtran}
\bibliography{references_new}

\end{document}

%% file: 1_introduction.tex
\section{Introduction}
\IEEEPARstart{R}{ecently}, there has been a renewed interest in risk-sensitive control for various applications, including robotics \cite{majumdar2020should, hammoud2021impedance}, remote state estimation \cite{sun2021remote}, and building evacuation \cite{barreiro2021risk}. 
A classical risk-sensitive control approach is to assess a random cost using the \emph{exponential utility} functional. This functional takes the form $\rho_\theta(G) := \frac{-2}{\theta} \log E(e^{\frac{-\theta}{2}G})$, where $\theta$ is a parameter and $G$ is a non-negative random variable. If $\theta < 0$, then large values of $G$ are exaggerated through the exponential transformation, which represents a risk-averse perspective. In contrast, the case of $\theta > 0$ represents a risk-seeking perspective. It can be shown that $\rho_\theta(G) \approx E(G) - \frac{\theta}{4}\text{var}(G)$ approximates a weighted sum of the mean and variance of $G$ under appropriate conditions, including $|\theta|$ being small \cite{whittle1981}.
Exponential-utility optimal control has been studied since the 1970s, and we first summarize early work.
%

In 1972, Howard and Matheson studied the optimization of an exponential utility criterion for a discrete-time Markov decision process (MDP) with finitely many states \cite{howardmat1972}. 
In 1973, Jacobson considered the problem of optimizing the exponential utility of a quadratic cost for a discrete-time linear system with Euclidean state and control spaces subject to additive Gaussian noise \cite{jacobson1973}. This problem is often called linear-exponential-quadratic-Gaussian (LEQG) or linear-exponential-quadratic-regulator (LEQR) control. LEQG theory was generalized by Whittle, in particular to the case of partially observed states, in the 1980s and 1990s; see \cite{whittle1991} and the references therein. 
Relations between controllers satisfying an $\mathcal{H}_\infty$-norm bound and the infinite-time LEQG controller were analyzed in the late 1980s \cite{gloverdoyle1988, glover1989minimum}. In 1999, di Masi and Stettner studied MDPs on continuous state spaces and discrete infinite-time horizons with exponential utility criteria \cite{masi1999}. In the early 2000s, relations between robust model predictive control (MPC) and MPC with an exponential utility criterion were investigated in the linear-quadratic setting \cite[Chap. 8.3]{minimaxmpc}. The exponential utility criterion belongs to the broader family of \emph{expected utility} criteria, which measure risk by transforming a random cost based on a user's subjective preferences.

Additional methods for quantifying and optimizing risk are presented by \cite[Chap. 6]{shapiro2009lectures}, \cite{ruszczynski2010risk}, for example. Specifically, Ruszczy{\'n}ski considered the optimization of a \emph{nested risk functional} for a discrete-time Borel-space MDP \cite{ruszczynski2010risk}, and related formulations have been studied, e.g., see \cite{shen2013risk, bauerle2021markov}. A nested risk functional takes the form $\rho_1(Z_1 + \rho_2(Z_2 + \cdots +  \rho_{N-1}(Z_{N-1} + \rho_N(Z_N))\cdots))$, where $Z_i$ is a random variable and $\rho_i$ is a mapping between spaces of random variables \cite{ruszczynski2010risk}. This functional is not straightforward to interpret, but it can be optimized using dynamic programming (DP) on the state space. In 2021, connections between nested risk functionals and distributionally robust MDPs were drawn \cite{bauerle2021markov}, and similar ideas were suggested earlier, e.g., see \cite[Eq. 4.11]{shapiro2012minimax}.

Another approach to risk-sensitive control is to optimize an expected cumulative cost subject to a risk constraint, for instance, see \cite{borkar2014risk, haskell2015convex, van2015distributionally, samuelson2018safety, tsiamis2020risk}. In particular, Tsiamis et al. considered a variance-like constraint in a linear-quadratic setting \cite{tsiamis2020risk}, whereas Refs. \cite{borkar2014risk, van2015distributionally, samuelson2018safety} considered \emph{Conditional Value-at-Risk} (CVaR) constraints. In contrast to expected utility criteria, CVaR is a quantile-based measure that quantifies an average cost in a fraction of worst cases. In prior work, we developed a controller using a CVaR objective \cite{chapman2021toward} and a safety analysis framework using CVaR \cite{mpctacsubmission}. Risk-constrained MDPs and MDPs with expected utility or CVaR criteria were studied by \cite{haskell2015convex} in an infinite-time setting using occupation measures and state space augmentation. State space augmentation is a technique for tracking history-dependent information that may be needed to characterize the stage-wise sub-problems of a multi-stage optimization problem. This technique is not needed when a sub-problem can be written in terms of the current state rather than the current and prior states.
More details about risk-sensitive MDPs can be found in \cite{bauerle2021markov}, for example.

While these maturing approaches to risk-sensitive control are intriguing, this paper concerns exponential utility, which remains popular in the literature. Specifically, the focus here is optimizing exponential utility for an MDP on a discrete finite-time horizon with Borel state and control spaces, which we call Problem A for convenience. Problem A permits non-Euclidean spaces, non-linear dynamics, non-quadratic costs, and non-Gaussian noise, and therefore generalizes the LEQG setting. In 2014, B{\"a}uerle and Rieder developed the state-of-the-art approach for solving Problem A and a broader class of MDP problems with expected utility criteria \cite{bauerlerieder}. However, the methodology employs an augmented state space, in which an extra state records the cumulative cost thus far \cite[Thm. 1, Cor. 1]{bauerlerieder}. Our contribution is to demonstrate that the additional complexity of state space augmentation is not required to solve Problem A, as we provide an alternate solution pathway using first principles from measure theory and real analysis.

There are key distinctions between our paper and the existing literature on MDPs with exponential utility criteria, which we outline below.
Many works concern MDPs on countable state spaces, see \cite{bielecki1999risk, hernandez1999existence, cavazos2010optimality, cavazos2011discounted, blancas2020discounted} 
for some examples, whereas we consider the more general setting of Borel spaces. As mentioned previously, the proof of \cite[Thm. 1, Cor. 1]{bauerlerieder} uses a Bellman equation that is defined on an augmented state space, whereas our approach does not use an augmented state space. Numerous papers about discrete-time Borel-space MDPs examine infinite-time settings, e.g., \cite{masi1999, di2000infinite, di2007infinite, jaskiewicz2007average, jaskiewicz2007note, asienkiewicz2017note, anantharam2017variational}, which naturally require different techniques compared to our work. di Masi and Stettner studied infinite-time problems in which the stage cost is continuous using a span contraction approach and two discounting approaches \cite{masi1999, di2000infinite}. Later, di Masi and Stettner generalized their work by approximating an MDP with one that is uniformly ergodic \cite{di2007infinite}. In contrast, our paper concerns a finite-time horizon and lower semi-continuous costs, and thus requires measure-theoretic arguments that are properly adapted to this setting. 
Ja{\'s}kiewicz and colleagues examined infinite-time problems using a vanishing discount factor approach \cite{jaskiewicz2007average}, MDPs in which the transition kernel only depends on the current control \cite{jaskiewicz2007note}, and later using the Banach fixed point theorem \cite{asienkiewicz2017note}. Anantharam and Borkar studied an infinite-time reward maximization problem using occupation measures \cite{anantharam2017variational}. While we focus on the discrete-time case, we note that continuous-time MDPs with exponential utility criteria have been examined by \cite{wei2016continuous, zhang2017continuous, pal2019risk, guo2019risk}, for example.

Notation: If $\mathcal{M}$ is a metrizable space, $\mathcal{B}_{\mathcal{M}}$ is the Borel sigma algebra on $\mathcal{M}$. $\mathcal{P}(D)$ is the set of probability measures on $(D,\mathcal{B}_{D})$ with the weak topology, where $D$ is a Borel space. Capital letters denote random objects, while lower-case letters denote the associated values; e.g., $x_t$ is a value of $X_t$. We define $\mathbb{T} := \{0,1,\dots,N-1\}$ and $\mathbb{T}_N := \{0,1,\dots,N\}$, where $N \in \mathbb{N}$ is given. $\mathbb{R}^* := \mathbb{R} \cup \{-\infty,+\infty\}$ is the extended real line. We abbreviate lower semi-continuous as lsc.

%% file: 2_problemstatement.tex
\section{Problem Statement}\label{secII}
Let $S$, $A$, and $D$ be Borel spaces of states, controls, and disturbances, respectively.
Consider a system on a discrete finite-time horizon of length $N \in \mathbb{N}$ of the form
\begin{equation}\label{dynamics}
    x_{t+1} = f_t(x_t, u_t, w_t) \; \; \; \forall t \in \mathbb{T},
\end{equation}
where $x_t \in S$, $u_t \in A$, and $w_t \in D$ are values of the random state $X_t$, the random control $U_t$, and the random disturbance $W_t$, respectively. The initial state $X_0$ is fixed at an arbitrary initial condition $x \in S$. The dynamics function $f_t : S \times A \times D \rightarrow S$ is Borel measurable. Given $(X_t,U_t)$, the disturbance $W_t$ is conditionally independent of $W_s$ for all $s \neq t$. The distribution of $W_t$, $p_t(\mathrm{d}w_t|x_t,u_t)$, is a Borel-measurable stochastic kernel on $D$ given $S \times A$. That is, the function $\gamma_t : S \times A \rightarrow \mathcal{P}(D)$ defined by $\gamma_t(x_t,u_t) := p_t(\mathrm{d}w_t|x_t,u_t)$ is Borel measurable. If $(x_t,u_t) \in S \times A$ is the value of $(X_t,U_t)$, then the distribution of $X_{t+1}$ is given by
\begin{equation}\label{qt}
    q_t(B|x_t,u_t) := p_t\big(\{w_t \in D : f_t(x_t,u_t,w_t) \in B \}\big|x_t,u_t\big)
\end{equation}
for all $B \in \mathcal{B}_S$ and $t \in \mathbb{T}$. We consider the class of deterministic Markov policies $\Pi$. Each $\pi \in \Pi$ takes the form $\pi = (\mu_0,\mu_1,\dots,\mu_{N-1})$ such that $\mu_t : S \rightarrow A$ is Borel measurable for each $t \in \mathbb{T}$.

We aim to define and optimize a random cost (where we specify the precise notion of optimality later). For this task, we are required to define a probability space $(\Omega,\mathcal{B}_\Omega,P_x^\pi)$, which is parametrized by an initial condition $x \in S$ and a policy $\pi \in \Pi$. The sample space $\Omega$ is defined by $\Omega := (S \times A)^N \times S$. That is, an element $\omega = (x_0,u_0,\dots,x_{N-1},u_{N-1},x_N) \in \Omega$ is a value of the random trajectory $(X_0,U_0,\dots,X_{N-1},U_{N-1},X_N)$. The coordinates of $\omega$ have causal dependencies due to the form of the dynamics \eqref{dynamics} and the class of policies $\Pi$. The distribution of the random trajectory is given by a probability measure $P_x^\pi$ on $(\Omega, \mathcal{B}_\Omega)$. The form of $P_x^\pi$ allows us to define a DP recursion for computing an optimal policy under certain conditions (to be specified).

Let $\delta_x$ denote the Dirac measure on $(S, \mathcal{B}_{S})$ concentrated at $x$. With slight abuse of notation, let $\delta_{\mu_t(x_t)}$ denote the Dirac measure on $(A,\mathcal{B}_A)$ concentrated at $\mu_t(x_t)$, where $x_t \in S$ is a value of $X_t$. Let $B \in \mathcal{B}_\Omega$ be a measurable rectangle, i.e., $B = B_{X_0} \times B_{U_0} \times B_{X_1} \times B_{U_1} \times \cdots \times B_{X_N}$, where $B_{X_i} \in \mathcal{B}_{S}$ for all $i \in \mathbb{T}_N$ and $B_{U_j} \in \mathcal{B}_A$ for all $j \in \mathbb{T}$. Then, we have \begin{equation}\label{pxpi} \begin{aligned} & P_x^\pi(B) \\ &=  \textstyle \int_{B_{X_0}} \int_{B_{U_0}} \int_{B_{X_1}} \int_{B_{U_1}} \cdots \int_{B_{X_N}} q_{N-1}(\mathrm{d}x_N|x_{N-1},u_{N-1}) \\ & \hphantom{=}\;\; \cdots  \delta_{\mu_1(x_1)}(\mathrm{d}u_1) \;  q_0(\mathrm{d}x_1|x_0,u_0) \; \delta_{\mu_0(x_0)}(\mathrm{d}u_0)  \; \delta_x(\mathrm{d}x_0). \end{aligned} \end{equation} 
The nested integrals should be taken from ``the inside to the outside.'' The integral with respect to $x_N \in S$ is taken over the set $B_{X_N}$; the integral with respect to $u_1 \in A$ is taken over $B_{U_1}$, etc. The reader may refer to \cite[Prop. C.10, Remark C.11, p. 178]{hernandez2012discrete} or \cite[Prop. 7.28, pp. 140--141]{bertsekas2004stochastic} for details. 

If $G :\Omega \rightarrow \mathbb{R}^*$ is Borel measurable and non-negative, then the expectation of $G$ with respect to $P_x^\pi$, $E_x^\pi(G)$, is given by
\begin{equation} \begin{aligned} 
& \textstyle \int_\Omega G(\omega)\; \mathrm{d} P_x^\pi(\omega) \\
   & = \textstyle \int_{S} \int_{A} \int_{S} \int_{A} \cdots \int_{S} G(x_0,u_0,\dots,x_{N-1},u_{N-1},x_N) \\
   & \hphantom{:=}\;\;
   q_{N-1}(\mathrm{d}x_N|x_{N-1},u_{N-1}) \cdots \delta_{\mu_1(x_1)}(\mathrm{d}u_1) \\ &\hphantom{:=}\; \; q_0(\mathrm{d}x_1|x_0,u_0) \; \delta_{\mu_0(x_0)}(\mathrm{d}u_0)  \; \delta_x(\mathrm{d}x_0). \end{aligned} 
\end{equation}  
$G$ is an extended random variable on $(\Omega,\mathcal{B}_\Omega,P_x^\pi)$ for each $x \in S$ and $\pi \in \Pi$.

We consider a particular random variable $Z$, representing a cost, that is incurred as the system operates over time. For any value $\omega = (x_0,u_0,\dots,x_{N-1},u_{N-1},x_N) \in \Omega$ of the random trajectory, we define
\begin{equation}\label{myZ}
    Z(\omega) := \textstyle \sum_{t = 0}^{N-1} c_t(x_t,u_t) + c_N(x_N),
\end{equation}
where $c_t : S \times A \rightarrow \mathbb{R}$ and $c_N : S \rightarrow \mathbb{R}$ are Borel measurable and bounded. In \eqref{myZ}, $x_t$ is the value of $X_t$ and $u_t$ is the value of $U_t$ associated with the trajectory value $\omega$.


A standard (risk-neutral) approach to manage $Z$ is to minimize its expectation $E_x^\pi(Z)$ over the class of policies $\Pi$. An alternative approach is to use the (risk-averse) \emph{exponential utility} functional. Let $\theta \in \Theta \subseteq (-\infty,0)$ be given, and define the optimal value function $V_\theta^* : S \rightarrow \mathbb{R}^*$ as follows:
\begin{equation}\label{vthetastar}
    V_\theta^*(x) := \inf_{\pi \in \Pi} \textstyle\frac{-2}{\theta}\log E_x^\pi\big(e^{\frac{-\theta}{2} Z}\big).
\end{equation}
If there is a policy $\pi_\theta^* \in \Pi$ such that $V_\theta^*(x) = \textstyle\frac{-2}{\theta}\log E_x^{\pi_\theta^*}\big(e^{\frac{-\theta}{2} Z}\big)$ for all $x \in S$, we say that $\pi_\theta^*$ is \emph{optimal for $V_\theta^*$}. The next assumption ensures that such a policy exists.
\begin{assumption}[Measurable Selection]\label{measselect}
We assume that \begin{enumerate}
\item $p_t(\mathrm{d}w_t|x_t,u_t)$ is a continuous stochastic kernel for all $t \in \mathbb{T}$. That is, the function $\gamma_t : S \times A \rightarrow \mathcal{P}(D)$ defined by $\gamma_t(x_t,u_t) := p_t(\mathrm{d}w_t|x_t,u_t)$ is continuous.
    \item $f_t$ is continuous for all $t \in \mathbb{T}$. $c_t$ is lower semi-continuous and bounded for all $t \in \mathbb{T}_N$.
    \item The set of controls $A$ is compact.
\end{enumerate}
\end{assumption}
\begin{remark}[Justification of Assumption \ref{measselect}]
Assumption \ref{measselect} is an example of a measurable selection condition. Such conditions are standard in stochastic control problems on Borel spaces, e.g., see \cite[Def. 8.7]{bertsekas2004stochastic}, \cite[Sec. 3.3]{hernandez2012discrete}. In a risk-neutral problem, which optimizes $E_x^\pi(Z)$, it is common to assume that $c_t$ is only bounded below. However, assuming that $c_t$ is bounded simplifies arguments for risk-sensitive MDPs, e.g., see \cite{bauerlerieder, haskell2015convex, asienkiewicz2017note}. $c_t$ being bounded ensures that $Z$ is bounded, which implies that $V_\theta^*$ is finite for all $\theta \in \Theta$.
\end{remark}

%% file: 3_algorithm.tex
\section{Dynamic Programming Algorithm}\label{secIII}
Next, we provide a DP algorithm for $V_\theta^*$.
\begin{algorithm}[DP for $V_\theta^*$]\label{valalgwhittle}
For any $\theta \in \Theta$, define $V_{t}^\theta : S \rightarrow \mathbb{R}^*$ recursively, $V_{N}^\theta(x) := c_N(x)$ and for $t = N-1, \dots,1, 0$,
\begin{subequations}\label{Vt}
\begin{equation}\label{10b}\begin{aligned}
& V_{t}^\theta(x) := {\underset{u \in A}\inf} v_{t+1}^\theta(x,u),
\end{aligned}\end{equation}
where $v_{t+1}^\theta : S \times A \rightarrow \mathbb{R}^*$ is defined by
\begin{equation}\label{mypsi}\begin{aligned}
v_{t+1}^\theta(x,u) &:= c_t(x,u) + \psi_t^{\theta}(x,u),\\
    \psi_t^{\theta}(x,u) & := \textstyle {\frac{-2}{\theta}} \log \textstyle (\int_{D} e^{\frac{-\theta}{2} V_{t+1}^\theta(f_t(x,u,w))} \; p_t(\mathrm{d}w|x,u)).
\end{aligned}\end{equation}
\end{subequations}
\end{algorithm}

Algorithm \ref{valalgwhittle} is a backwards recursion that applies an exponential transformation to the cost-to-go $V_{t+1}^\theta$ and resembles existing formulations, e.g., see \cite[Eq. 2.3]{masi1999}, \cite[Remark 1]{bauerlerieder}, and \cite[Eq. 1.1]{asienkiewicz2017note}. Our contribution is not the algorithm itself but instead the direct pathway that we follow to solve the risk-averse control problem of interest. Namely, we use first principles from real analysis and measure theory and build on arguments from \cite{bertsekas2004stochastic, ash1972} to prove two theorems. Theorem \ref{lscremark} shows that $V_{t}^\theta$ is lower semi-continuous (lsc) and bounded, and there is a Borel-measurable function $\hspace{-.5mm}\mu_{t}^\theta \hspace{-.5mm} : \hspace{-.5mm}S\hspace{-.5mm} \rightarrow\hspace{-.5mm} A\hspace{-.5mm}$ such that
\begin{equation}\label{existinf}
   V_{t}^\theta(x) = v_{t+1}^\theta(x,\mu_{t}^\theta(x)) \; \; \; \forall x \in S.
\end{equation}
In Theorem \ref{optimality}, we show that
the (non-unique) policy $\pi_\theta^*:= (\mu_{0}^\theta, \mu_{1}^\theta, \dots, \mu_{N-1}^\theta)$ is optimal for $V_\theta^*$ and $V_{0}^\theta =  V_\theta^*$.


%% file: Appendix.tex
\section{Analysis of Dynamic Programming Iterates}
In this section, we first prove Proposition \ref{prop1}, which shows that the composition of a continuous increasing function and a bounded lsc function is lsc. Then, we use Proposition \ref{prop1} and another preliminary result (Lemma \ref{phiprimeislsc}, Appendix) to analyze the DP iterates $V_0^\theta, V_1^\theta, \dots, V_N^\theta$ in Theorem \ref{lscremark}.

\begin{proposition}[Comp. of con't, inc. and lsc]\label{prop1}
Let $\mathcal{M}$ be a metrizable space. Assume that $\mathcal{Y}_i := (a_i,b_i) \subseteq \mathbb{R}$ is non-empty for $i = 1,2$. Suppose that $\kappa_1 : \mathcal{Y}_1 \rightarrow \mathcal{Y}_2$ is continuous and increasing, and $\kappa_2 : \mathcal{M} \rightarrow \mathcal{Y}_1$ is lsc and bounded. (There are scalars $\underline{c}$ and $\overline{c}$ such that $[\underline{c}, \overline{c}] \subset \mathcal{Y}_1$ and $\underline{c} \leq \kappa_2(y) \leq \overline{c}$ for all $y \in \mathcal{M}$.) Then, $\kappa_1 \circ \kappa_2 : \mathcal{M} \rightarrow \mathcal{Y}_2$ is lsc.
\end{proposition}
\begin{proof}
To show that $\kappa_1 \circ \kappa_2$ is lsc, we must show that
   $ \liminf_{i \rightarrow \infty}  \kappa_1\big(\kappa_2(x^i)\big) \geq \kappa_1\big(\kappa_2(x)\big)$,
where $\{x^i\}_{i = 1}^\infty$ is a sequence in $\mathcal{M}$ converging to $x \in \mathcal{M}$.\footnote{A key aspect of the proof of Proposition \ref{prop1} is the use of the bounds $\underline{c}$ and $\overline{c}$ to guarantee that a limit inferior is in the domain of $\kappa_1$. This and the continuity of $\kappa_1$ allow us to exchange the order of a limit and $\kappa_1$.} Since $\{x^i\}_{i = 1}^\infty$ converges to $x$ and $\kappa_2$ is lsc, it holds that
\begin{equation}\label{my717171}
  \lim_{i \rightarrow \infty} \inf_{k \geq i} \kappa_2(x^k) := \liminf_{i \rightarrow \infty} \kappa_2(x^i) \geq \kappa_2(x).
\end{equation}
Since $\kappa_2$ is bounded below by $\underline{c}$ and above by $\overline{c}$, we have
\begin{equation}
\underline{c}   \leq \inf_{k \geq i} \kappa_2(x^k) \leq \inf_{k \geq i+1} \kappa_2(x^k) \leq \overline{c} \; \; \; \forall i \in \mathbb{N},
\end{equation}
which implies that
  $\underline{c} \leq \lim_{i \rightarrow \infty} \inf_{k \geq i} \kappa_2(x^k) \leq \overline{c}$. 
Since $\{\inf_{k \geq i} \kappa_2(x^k)\}_{i=1}^\infty$ is a sequence in $[\underline{c}, \overline{c}]$ which converges to a point in $[\underline{c}, \overline{c}]$, $[\underline{c}, \overline{c}]$ is a non-empty subset of $\mathcal{Y}_1$, and $\kappa_1$ is continuous on $\mathcal{Y}_1$, we find that
\begin{equation}\label{my7474}
   \kappa_1 \Big( \lim_{i \rightarrow \infty} \inf_{k \geq i} \kappa_2(x^k) \Big) = \lim_{i \rightarrow \infty} \kappa_1 \Big( \inf_{k \geq i} \kappa_2(x^k) \Big).
\end{equation}
Since $\kappa_1$ is increasing and by \eqref{my717171}, we have
\begin{equation}\label{my757575}
 \kappa_1 \Big( \lim_{i \rightarrow \infty} \inf_{k \geq i} \kappa_2(x^k) \Big)  \geq \kappa_1 ( \kappa_2(x) ).
\end{equation}
Moreover, since $\kappa_1$ is increasing, for any $i \in \mathbb{N}$, it holds that
\begin{equation}
   \forall k \geq i, \; \; \; \kappa_1(\kappa_2(x^k)) \geq \kappa_1 \Big( \inf_{k \geq i} \kappa_2(x^k) \Big).
\end{equation}
Thus, $\kappa_1 \big( \inf_{k \geq i} \kappa_2(x^k) \big) \in \mathbb{R}$ is a lower bound for the set $\big\{ \kappa_1(\kappa_2(x^k)) : k \geq i\big\}$, which implies that
\begin{equation}\label{my7777}
    \inf_{k \geq i} \kappa_1(\kappa_2(x^k))  \geq \kappa_1 \Big( \inf_{k \geq i} \kappa_2(x^k) \Big).
\end{equation}
By letting $i$ tend to infinity, it holds that
\begin{equation}
    \lim_{i \rightarrow \infty}  \inf_{k \geq i} \kappa_1(\kappa_2(x^k)) \geq \lim_{i \rightarrow \infty} \kappa_1 \Big( \inf_{k \geq i} \kappa_2(x^k) \Big),
\end{equation}
which is equivalent to
\begin{equation}
\liminf_{i \rightarrow \infty} \kappa_1(\kappa_2(x^i)) \geq \kappa_1 \Big( \lim_{i \rightarrow \infty} \inf_{k \geq i} \kappa_2(x^k) \Big)
\end{equation}
by \eqref{my7474}. Finally, by \eqref{my757575}, we derive the desired result.
\end{proof}

Next, we use Proposition \ref{prop1} to prove Theorem \ref{lscremark}.
\begin{theorem}[Properties of $V_{t}^\theta$]\label{lscremark}
Assume Assumption \ref{measselect}. For all $t \in \mathbb{T}_N$, $V_{t}^\theta$ is lsc and bounded. For all $t \in  \mathbb{T}$, there is a Borel-measurable function $\mu_{t}^\theta : S \rightarrow A$ such that \eqref{existinf} holds. 
\end{theorem}
\begin{proof} By induction. For brevity, denote $\mathcal{S} := S \times A \times D$. $V_N^\theta = c_N$ is lsc and bounded by Assumption \ref{measselect}. Now, suppose (the induction hypothesis) that for some $t \in \mathbb{T}$, $V_{t+1}^\theta$ is lsc and bounded. The key step is to show that $\psi_t^\theta$ \eqref{mypsi} is lsc and bounded.\footnote{If $\psi_t^\theta$ is lsc and bounded, then $v_{t+1}^\theta = c_t + \psi_t^\theta$ is lsc and bounded because the sum of two lsc and bounded functions is lsc and bounded. Since $v_{t+1}^\theta$ is bounded, $V_{t}^\theta$ is bounded \eqref{10b}. The remaining desired conclusions follow from a known result, which we describe next. Since $v_{t+1}^\theta : S \times A \rightarrow \mathbb{R}$ is lsc, $A$ is compact, and $V_{t}^\theta(x) = \inf_{u \in A} v_{t+1}^\theta(x,u)$ $\forall x \in S$, it holds that $V_{t}^\theta$ is lsc, and there is a Borel-measurable function $\mu_t^\theta : S \rightarrow A$ such that $V_{t}^\theta(x) = v_{t+1}^\theta(x,\mu_{t}^\theta(x))$ $\forall x \in S$
by a special case of \cite[Prop. 7.33, p. 153]{bertsekas2004stochastic}. In summary, if $\psi_t^\theta$ is lsc and bounded, then a) $V_{t}^\theta$ is lsc and bounded and b) a Borel-measurable function $\mu_t^\theta$ satisfying \eqref{existinf} exists. This logic repeats backwards in time to complete the proof.} Boundedness of $\psi_t^\theta$ follows from boundedness of $V_{t+1}^\theta$. For showing boundedness, note that the function $\phi_t : \mathcal{S} \rightarrow \mathbb{R}$ defined by
\begin{equation}\label{myphi}
   \phi_t(x,u,w) := e^{\frac{-\theta}{2}V_{t+1}^\theta(f_t(x,u,w))}
\end{equation}
is non-negative and bounded. (We drop the superscript $\theta$ on the left-hand-side for brevity.) Also, for any $(x,u) \in S \times A$, the function $\phi_t(x,u,\cdot) : D \rightarrow \mathbb{R}$ is Borel measurable because it is a composition of Borel-measurable functions. In addition, since $p_t(\mathrm{d}w|x,u)$ is a probability measure on $(D, \mathcal{B}_{D})$, it follows that the (Lebesgue) integral
\begin{equation}\label{myphitprime}
    \phi_t'(x,u) := \textstyle \int_{D} \phi_t(x,u,w) \; p_t(\mathrm{d}w|x,u)
\end{equation}
exists and is finite for all $(x,u) \in S \times A$. 

To complete the proof, we show that $\psi_t^\theta$ \eqref{mypsi} is lsc. Since $f_t$ is continuous by Assumption \ref{measselect}, $V_{t+1}^\theta$ is lsc by the induction hypothesis, and $\frac{-\theta}{2} > 0$, the function $g_t : \mathcal{S} \rightarrow \mathbb{R}$ defined by
\begin{equation}\label{Myg}
 \textstyle   g_t(x,u,w) := \frac{-\theta}{2}V_{t+1}^\theta(f_t(x,u,w))
\end{equation}
is lsc. Indeed, let $\{(x^i,u^i,w^i)\}_{i=1}^\infty$ be a sequence in $\mathcal{S}$ converging to $(x,u,w) \in \mathcal{S}$. To show that $V_{t+1}^\theta \circ f_t$ is lsc, we must prove that
\begin{equation}
    \liminf_{i \rightarrow \infty} V_{t+1}^\theta(f_t(x^i,u^i,w^i)) \geq V_{t+1}^\theta(f_t(x,u,w)).
\end{equation}
Since $\{(x^i,u^i,w^i)\}_{i=1}^\infty$ converges to $(x,u,w)$ and $f_t$ is continuous, $\{f_t(x^i,u^i,w^i)\}_{i=1}^\infty$ converges to $f_t(x,u,w)$. Since the latter is a converging sequence in $S$ and $V_{t+1}^\theta : S \rightarrow \mathbb{R}$ is lsc, we have
\begin{equation}\label{466}
    \liminf_{i \rightarrow \infty} V_{t+1}^\theta(f_t(x^i,u^i,w^i)) \geq V_{t+1}^\theta(f_t(x,u,w)),
\end{equation}
which proves that the composition $V_{t+1}^\theta \circ f_t$ is lsc. Since $V_{t+1}^\theta \circ f_t$ is lsc and $\frac{-\theta}{2} > 0$, the function $g_t = \frac{-\theta}{2}V_{t+1}^\theta \circ f_t$ is lsc because multiplying \eqref{466} by a positive constant preserves the direction of the inequality. 

By Proposition \ref{prop1}, it holds that $\phi_t = \exp \circ g_t$ \eqref{myphi} is lsc because $\exp : \mathbb{R} \rightarrow (0,+\infty)$ is continuous and increasing and $g_t: \mathcal{S} \rightarrow \mathbb{R}$ \eqref{Myg} is lsc and bounded. In addition, $\phi_t$ is bounded as a consequence of $g_t$ being bounded.

It follows that $\phi_t'$ \eqref{myphitprime} is bounded and lsc. The latter property holds in particular because $\phi_t$ \eqref{myphi} is lsc and $p_t(\mathrm{d}w|x,u)$ is a continuous stochastic kernel (see Lemma \ref{phiprimeislsc}, Appendix). 

We use Proposition \ref{prop1} again to conclude that $\log \circ \; \phi_t' : S \times A \rightarrow \mathbb{R}$ is lsc. To apply Proposition \ref{prop1}, we choose $\mathcal{M} = S \times A$, $\mathcal{Y}_1 = (0,+\infty)$, $\mathcal{Y}_2 = \mathbb{R}$, $\kappa_1 = \log$, $\kappa_2 = \phi_t'$, and $[\underline{c},\overline{c}] = [e^{\frac{-\theta}{2}\underline{b}},e^{\frac{-\theta}{2}\overline{b}}] \subset \mathcal{Y}_1$, where $\underline{b}$ is a lower bound and $\overline{b}$ is an upper bound for $V_{t+1}^\theta$. 
Finally, since $\log \circ \; \phi_t'$ is lsc and $\frac{-2}{\theta} > 0$, we conclude that $\psi_t^\theta = \frac{-2}{\theta} \log \circ \; \phi_t'$ is lsc.
\end{proof}

By proving Theorem \ref{lscremark}, we guarantee that the functions  $V_0^\theta, V_1^\theta, \dots, V_N^\theta$ satisfy properties which facilitate the optimality result (Theorem \ref{optimality}) in the following section.
%
%
\section{Existence of an Optimal Risk-Averse Policy}
In this section, first we prove a DP recursion for the risk-averse control problem; the sum-to-product property of the exponential function is particularly useful for this proof (Lemma \ref{dynprogremark}). Then, we use Lemma \ref{dynprogremark} and Theorem \ref{lscremark} to show the equality $V_{0}^\theta = V_{\theta}^*$ and to construct an optimal risk-averse policy (Theorem \ref{optimality}). 

Define the random cost-to-go $Z_t$ for time $t \in \mathbb{T}_N$ as follows: for all $\omega = (x_0,u_0,\dots,x_{N-1},u_{N-1},x_N) \in \Omega$, 
\begin{equation}\label{myZt}
  \textstyle  Z_t(\omega) := \begin{cases} c_N(x_N) + \sum_{i=t}^{N-1} c_i(x_i,u_i) &  \text{if } t \in \mathbb{T} \\ c_N(x_N) & \text{if } t = N\end{cases}.
\end{equation}
Note that $Z_t(\omega) = c_t(x_t,u_t) + Z_{t+1}(\omega)$ for any $t \in \mathbb{T}$ and $\omega \in \Omega$ of the form specified above, and $Z_0 = Z$ \eqref{myZ}. While $Z_t$ is a random variable on $\Omega$, the right-hand-side of \eqref{myZt} does not depend on the trajectory prior to time $t$. This is useful for deriving a recursion that is history-dependent only through the current state. For $t \in \mathbb{T}_N$, $x_t \in S$, $\theta \in \Theta$, and $\pi \in \Pi$, we denote a conditional expectation of $e^{\frac{-\theta}{2} Z_t}$ by
  $ W_t^{\pi,\theta}(x_t) = E^\pi(e^{\frac{-\theta}{2} Z_t}| X_t = x_t )$. We define an induced probability measure $P_{X_t}^\pi(B) := P_x^\pi\big(\{X_t \in B \} \big)$ for $B \in \mathcal{B}_{S}$.
We use the following abbreviations: a.e. = almost every and w.r.t. = with respect to.

\begin{lemma}[DP recursion]\label{dynprogremark}
Let $\theta \in \Theta$ and $\pi = (\mu_0,\mu_1,\dots,\mu_{N-1}) \in \Pi$ be given. Under Assumption \ref{measselect}, it holds that $W_t^{\pi,\theta}(x_t)\in (0,+\infty)$ and
\begin{equation*}\begin{aligned}
& W_t^{\pi,\theta}(x_t) =\\
 &  \textstyle e^{\frac{-\theta}{2}c_t(x_t,\mu_t(x_t))}  \int_{D} W_{t+1}^{\pi,\theta}(f_t(x_t,\mu_t(x_t),w)) \; p_t(\mathrm{d}w|x_t,\mu_t(x_t))
\end{aligned}\end{equation*}
for almost every $x_t \in S$ with respect to $P_{X_t}^\pi$ and for all $t \in \mathbb{T}$. 
\end{lemma}
\begin{proof}
To derive the form of $W_t^{\pi,\theta}$, one applies the definition of conditional expectation \cite[Th. 6.3.3, p. 245]{ash1972}. It follows that the function $W_t^{\pi,\theta} : S \rightarrow \mathbb{R}^*$ is Borel measurable and $W_t^{\pi,\theta}(x_t)$ equals 
\begin{equation}\label{85}\begin{aligned}
&      \textstyle  \int_A \int_{S} \int_A \cdots \int_{A} \int_{S} e^{\frac{-\theta}{2} (c_N(x_N) + \sum_{i=t}^{N-1} c_i(x_i,u_i))}\\ 
  & \hphantom{=} \;\; q_{N-1}(\mathrm{d}x_{N}|x_{N-1},u_{N-1}) \; \delta_{\mu_{N-1}(x_{N-1})}(\mathrm{d}u_{N-1})  \cdots \\ & \hphantom{=} \;\; \delta_{\mu_{t+1}(x_{t+1})}(\mathrm{d}u_{t+1}) \; q_{t}(\mathrm{d}x_{t+1}|x_t,u_t) \; \delta_{\mu_t(x_t)}(\mathrm{d}u_t) \end{aligned}
\end{equation}
for a.e. $x_t \in S$ w.r.t. $P_{X_t}^\pi$ and for all $t \in \mathbb{T}$.
In addition, the function $W_N^{\pi,\theta} : S \rightarrow \mathbb{R}^*$ is Borel measurable and satisfies
     $W_N^{\pi,\theta}(x_N)  =  e^{\frac{-\theta}{2} c_N(x_N)}$ for a.e. $x_N \in S$ w.r.t. $P_{X_N}^\pi$.
Since $c_t$ is bounded for all $t \in \mathbb{T}_N$ and the exponential is positive, we have that $W_t^{\pi,\theta}(x_t) \in (0,+\infty)$ for a.e. $x_t \in S$ w.r.t. $P_{X_t}^\pi$ and for all $t \in \mathbb{T}_N$. Details about applying \cite[Th. 6.3.3]{ash1972} are provided in a footnote.\footnote{To apply \cite[Th. 6.3.3]{ash1972}, note that $e^{\frac{-\theta}{2} Z_t}$ is a random variable on $(\Omega,\mathcal{B}_{\Omega},P_x^\pi)$ for any $x \in S$ and $\pi \in \Pi$. $X_t$ is a random object. 
The expectation $E_x^\pi(e^{\frac{-\theta}{2} Z_t}) := \int_{\Omega} e^{\frac{-\theta}{2} Z_t(\omega)} \mathrm{d}P_x^\pi(\omega)$ exists because $e^{\frac{-\theta}{2} Z_t(\omega)} \geq 0$ for all $\omega \in \Omega$. }

Let $t \in \{0,1,\dots,N-2\}$ be given. 
Since $e^{\frac{-\theta}{2}c_t(x_t,u_t)}$ does not depend on the trajectory after time $t$, it can be placed ``outside'' several integrals so that \eqref{85} becomes
%
 \begin{equation}\label{89}\begin{aligned} 
    & \textstyle  W_t^{\pi,\theta}(x_t)  \\
&    =   \textstyle  \int_A e^{\frac{-\theta}{2}c_t(x_t,u_t)} \int_{S} \int_A \cdots \int_{A} \int_{S} e^{\frac{-\theta}{2} (c_N(x_N) + \sum_{i=t+1}^{N-1} c_i(x_i,u_i))}\\ 
  & \hphantom{=} \; q_{N-1}(\mathrm{d}x_{N}|x_{N-1},u_{N-1}) \; \delta_{\mu_{N-1}(x_{N-1})}(\mathrm{d}u_{N-1})  \cdots \\ & \hphantom{=} \; \delta_{\mu_{t+1}(x_{t+1})}(\mathrm{d}u_{t+1}) \; q_{t}(\mathrm{d}x_{t+1}|x_t,u_t) \; \delta_{\mu_{t}(x_{t})}(\mathrm{d}u_{t}) 
\end{aligned}
\end{equation}
%
for a.e. $x_t \in S$ w.r.t. $P_{X_t}^\pi$. Since $t+1 \in \mathbb{T}$, by applying \eqref{85}, we have that
\begin{equation}\label{90}\begin{aligned}
    & \textstyle  W_{t+1}^{\pi,\theta}(x_{t+1}) \\ & =  \textstyle  \int_A \cdots \int_{A} \int_{S} e^{\frac{-\theta}{2} (c_N(x_N) + \sum_{i=t+1}^{N-1} c_i(x_i,u_i))}\\ 
  & \hphantom{==} q_{N-1}(\mathrm{d}x_{N}|x_{N-1},u_{N-1}) \; \delta_{\mu_{N-1}(x_{N-1})} (\mathrm{d}u_{N-1}) \\ & \hphantom{==} \cdots \delta_{\mu_{t+1}(x_{t+1})} (\mathrm{d}u_{t+1}) 
\end{aligned}\end{equation}
for a.e. $x_{t+1} \in S$ w.r.t. $P_{X_{t+1}}^\pi$. 
The expression for $W_{t+1}^{\pi,\theta}(x_{t+1})$ \eqref{90} appears in \eqref{89}, from which we conclude that
%
\begin{equation}\label{114}\begin{aligned}
 & W_t^{\pi,\theta}(x_t)   = \\ & \textstyle  \int_A e^{\frac{-\theta}{2}c_t(x_t,u_t)} 
  \int_{S}  W_{t+1}^{\pi,\theta}(x_{t+1}) \; q_{t}(\mathrm{d}x_{t+1}|x_t,u_t) \; \delta_{\mu_t(x_t)}(\mathrm{d}u_t)
\end{aligned}\end{equation}
for a.e. $x_t \in S$ w.r.t. $P_{X_t}^\pi$.
By using the definition of the Dirac measure $\delta_{\mu_{t}(x_{t})}$ and the definition of $q_t$ \eqref{qt}, we complete the derivation of the recursion for $t \in \{0,1,\dots,N-2\}$. The derivation for $t = N-1$ is analogous.
\end{proof}
We would like to provide further explanation for \eqref{114}, as there are subtle arguments involving the almost-everywhere notions. Let us define
\begin{equation}\label{myvarphi}
\begin{aligned}
    & \varphi_{t+1}^{\pi,\theta}(x_{t+1}) \\ & \coloneqq \textstyle  \int_A \cdots \int_{A} \int_{S} e^{\frac{-\theta}{2} (c_N(x_N) + \sum_{i=t+1}^{N-1} c_i(x_i,u_i))}\\ 
  & \hphantom{==} q_{N-1}(\mathrm{d}x_{N}|x_{N-1},u_{N-1}) \; \delta_{\mu_{N-1}(x_{N-1})} (\mathrm{d}u_{N-1}) \\ & \hphantom{==} \cdots \delta_{\mu_{t+1}(x_{t+1})} (\mathrm{d}u_{t+1})
\end{aligned}
\end{equation}
for all $x_{t+1} \in S$. That is, $\varphi_{t+1}^{\pi,\theta}(x_{t+1}) = W_{t+1}^{\pi,\theta}(x_{t+1})$ for a.e. $x_{t+1} \in S$ w.r.t. $P_{X_{t+1}}^\pi$. From \eqref{89} and \eqref{myvarphi},
\begin{equation}
    \begin{aligned}
     &    \textstyle  W_t^{\pi,\theta}(x_t) \\  
  &  =   \textstyle e^{\frac{-\theta}{2}c_t(x_t,\mu_{t}(x_{t}))} \int_{S}  \varphi_{t+1}^{\pi,\theta}(x_{t+1}) \; q_{t}(\mathrm{d}x_{t+1}|x_t,\mu_{t}(x_{t}))
    \end{aligned}
\end{equation}
for a.e. $x_t \in S$ w.r.t. $P_{X_{t}}^\pi$. Denote $q_{t}^{\mu_t}(\cdot|x_t) \coloneqq q_{t}(\cdot|x_t,\mu_{t}(x_{t}))$ for brevity. To guarantee \eqref{114}, we need
\begin{equation}\label{my2929}
   \textstyle \int_{S}  \varphi_{t+1}^{\pi,\theta} \; \mathrm{d}q_{t}^{\mu_t}(\cdot|x_t) = \int_{S}  W_{t+1}^{\pi,\theta} \; \mathrm{d}q_{t}^{\mu_t}(\cdot|x_t)
\end{equation}
for a.e. $x_t \in S$ w.r.t. $P_{X_{t}}^\pi$. Now, for every $\underline{S} \in \mathcal{B}_{S}$,
\begin{equation}\label{30}
\begin{aligned}
    P_{X_{t+1}}^\pi(\underline{S}) & = \textstyle \int_S \int_A  q_t(\underline{S}|x_t,u_t) \; \delta_{\mu_t(x_t)}(\mathrm{d}u_t) \; \mathrm{d}P_{X_t}^\pi(x_t) \\
    & = \textstyle \int_S q_t(\underline{S}|x_t,\mu_t(x_t))  \; \mathrm{d}P_{X_t}^\pi(x_t) \\
    & = \textstyle \int_S q_{t}^{\mu_t}(\underline{S}|x_t)  \; \mathrm{d}P_{X_t}^\pi(x_t).
    \end{aligned}
\end{equation}
For brevity, let us denote
\begin{equation}\label{311}
    S_{t+1}^{\pi,\theta} \coloneqq \left\{z \in S : \varphi_{t+1}^{\pi,\theta}(z) \neq W_{t+1}^{\pi,\theta}(z) \right\}.
\end{equation}
Since $\varphi_{t+1}^{\pi,\theta} = W_{t+1}^{\pi,\theta}$ a.e. w.r.t. $P_{X_{t+1}}^\pi$ and by \eqref{30}--\eqref{311},
\begin{equation}\label{31}
    P_{X_{t+1}}^\pi(S_{t+1}^{\pi,\theta}) = \textstyle \int_S q_{t}^{\mu_t}(S_{t+1}^{\pi,\theta}|x_t)  \; \mathrm{d}P_{X_t}^\pi(x_t) = 0.
\end{equation}
Since $q_{t}^{\mu_t}(S_{t+1}^{\pi,\theta}|\cdot)$ is a non-negative Borel-measurable function on $S$ and \eqref{31} holds, it follows that $q_{t}^{\mu_t}(S_{t+1}^{\pi,\theta}|x_t) = 0$ for a.e. $x_t \in S$ w.r.t. $P_{X_t}^\pi$ by \cite[Th. 1.6.6 (b)]{ash1972}. 
The last statement is equivalent to
\begin{equation}
    q_{t}^{\mu_t}\left(\left\{z \in S : \varphi_{t+1}^{\pi,\theta}(z) = W_{t+1}^{\pi,\theta}(z) \right\} \Big|x_t \right) = 1
\end{equation}
for a.e. $x_t \in S$ w.r.t. $P_{X_t}^\pi$. That is, $\varphi_{t+1}^{\pi,\theta} = W_{t+1}^{\pi,\theta}$ almost everywhere with respect to $q_{t}^{\mu_t}(\cdot|x_t)$ for almost every $x_t \in S$ with respect to $P_{X_t}^\pi$. Finally, by another classic integration theorem \cite[Th. 1.6.5 (b)]{ash1972}, 
we conclude that \eqref{my2929} holds for almost every $x_t \in S$ with respect to $P_{X_t}^\pi$.

The last result proves optimality.
\begin{theorem}[Optimality of $V_{0}^\theta$ and $\pi_\theta^*$]\label{optimality}
Under Assumption \ref{measselect}, 
it holds that
    $V_{0}^\theta(x) =  V_\theta^*(x) = \textstyle\frac{-2}{\theta}\log E_x^{\pi_\theta^*}\big(e^{\frac{-\theta}{2} Z}\big)$ for all $x \in S$, where $\pi_\theta^* := (\mu_0^\theta,\mu_1^\theta, \dots,\mu_{N-1}^\theta)$ is non-unique and is given by Theorem \ref{lscremark}.
\end{theorem}
\begin{proof}
Let $x \in S$ be given. Note that
     $W_0^{\pi,\theta}(x)  = E_x^\pi(e^{\frac{-\theta}{2} Z})$ for all $\pi \in \Pi$,
and recall that $Z = Z_0$. It suffices to show the following two statements:
\begin{align}\label{toshow84}
  \textstyle \frac{-2}{\theta} \log W_t^{\pi,\theta}(x_t) & \geq V_{t}^\theta(x_t)
\end{align}
for a.e. $x_t \in S$ w.r.t. $P_{X_t}^\pi$, for all $\pi \in \Pi$, and for all $t \in \mathbb{T}_N$, and
\begin{equation}\label{toshownext}
      \textstyle \frac{-2}{\theta} \log W_t^{\pi_\theta^*,\theta}(x_t) = V_{t}^\theta(x_t)
\end{equation}
for a.e. $x_t \in S$ w.r.t. $P_{X_t}^{\pi_\theta^*}$ and for all $t \in \mathbb{T}_N$. (Consider $t=0$, use $P_{X_0}^\pi = P_x^\pi \circ X_{0}^{-1}$ in the change-of-measure theorem \cite[Th. 1.6.12, p. 50]{ash1972}, and note that the realizations of $X_0$ are concentrated at $x$.) 

We proceed by induction. For the base case, we have
   $ \frac{-2}{\theta} \log W_N^{\pi,\theta}(x_N)  = \frac{-2}{\theta} \log \big(e^{\frac{-\theta}{2} c_N(x_N)}\big)
   =  V_{N}^\theta(x_N)$
for a.e. $x_N \in S$ w.r.t. $P_{X_N}^\pi$ and for all $\pi \in \Pi$. 

Now, assume the induction hypothesis for \eqref{toshow84}: for some $t \in \mathbb{T}$, it holds that
$\textstyle \frac{-2}{\theta} \log W_{t+1}^{\pi,\theta}(x_{t+1}) \geq V_{t+1}^\theta(x_{t+1})$ 
for a.e. $x_{t+1} \in S$ w.r.t. $P_{X_{t+1}}^\pi$ and for all $\pi \in \Pi$. 
Since $\frac{-\theta}{2} > 0$, the exponential is increasing, and $e^{\log a} = a$ for all $a \in (0,+\infty)$, the induction hypothesis is equivalent to
\begin{align}
     \textstyle W_{t+1}^{\pi,\theta}(x_{t+1}) & \geq e^{\frac{-\theta}{2} V_{t+1}^\theta(x_{t+1})} \label{94c}
\end{align}
for a.e. $x_{t+1} \in S$ w.r.t. $P_{X_{t+1}}^\pi$ and for all $\pi \in \Pi$. Let $\pi = (\mu_0,\mu_1,\dots,\mu_{N-1}) \in \Pi$ be given. We use the recursion provided by Lemma \ref{dynprogremark}, the inequality in \eqref{94c}, $V_{t+1}^\theta$ being lsc and bounded below (Theorem \ref{lscremark}), and $W_{t+1}^{\pi,\theta}$ being Borel measurable 
to derive the inequality: $W_t^{\pi,\theta}(x_t) \geq$
\begin{equation}\label{mymy46}\begin{aligned}
 & \textstyle e^{\frac{-\theta}{2}c_t(x_t,\mu_t(x_t))}   \int_{D} e^{\frac{-\theta}{2} V_{t+1}^\theta(f_t(x_t,\mu_t(x_t),w))} \; p_t(\mathrm{d}w|x_t,\mu_t(x_t))
\end{aligned}\end{equation}
for a.e. $x_t \in S$ w.r.t. $P_{X_t}^\pi$. 
The right-hand-side of the inequality 
is a product of elements of $(0,+\infty)$ in particular because $V_{t+1}^\theta$ is bounded (Theorem \ref{lscremark}). Since $\log(a b) = \log a + \log b$ for all $a \in (0,+\infty)$ and $b \in (0,+\infty)$, and the natural logarithm is increasing, 
we derive the next inequality: $\log W_t^{\pi,\theta}(x_t) \geq$
\begin{equation}\begin{aligned}
    &  \textstyle \log e^{\frac{-\theta}{2}c_t(x_t,\mu_t(x_t))} \\ & \hphantom{\geq} \; + \textstyle \log \big(\int_{D} e^{\frac{-\theta}{2} V_{t+1}^\theta(f_t(x_t,\mu_t(x_t),w))} \; p_t(\mathrm{d}w|x_t,\mu_t(x_t))\big)
\end{aligned}\end{equation}
for a.e. $x_t \in S$ w.r.t. $P_{X_t}^\pi$.
By simplifying the first term in the sum and multiplying by $\frac{-2}{\theta}> 0$, it follows that $\frac{-2}{\theta} \log W_t^{\pi,\theta}(x_t) \geq v_{t+1}^\theta(x_t,\mu_t(x_t))$ for a.e. $x_t \in S$ w.r.t. $P_{X_t}^\pi$, where $v_{t+1}^\theta$ is defined by \eqref{mypsi}.
Moreover, since $v_{t+1}^\theta(x_t,\mu_t(x_t)) \geq \inf_{u \in A}v_{t+1}^\theta(x_t,u) = V_t^\theta(x_t)$ \eqref{10b} for every $x_t \in S$, we conclude that $\frac{-2}{\theta} \log W_t^{\pi,\theta}(x_t) \geq V_{t}^\theta(x_t)$ for a.e. $x_t \in S$ w.r.t. $P_{X_t}^\pi$. This proves the induction step for \eqref{toshow84}.

A similar induction argument proves that \eqref{toshownext} holds for a.e. $x_t \in S$ w.r.t. $P_{X_t}^{\pi_\theta^*}$ and for all $t \in \mathbb{T}_N$. In particular, this argument uses Theorem \ref{lscremark} to guarantee that \eqref{existinf} holds. 
\end{proof}

%% file: 6_conclusions.tex
\section{Concluding Remarks}\label{secVI}
Here, we have studied a classical risk-averse control problem for an MDP with Borel state and control spaces on a discrete finite-time horizon, where risk is characterized using the exponential utility functional. While exponential-utility optimal control is well-understood in settings with linear-quadratic assumptions, countable state spaces, or infinite-time horizons, it is not well-understood outside of these settings from a basic analytical perspective. Using first principles from measure theory and real analysis, we have presented a more basic path to the solution in comparison to the existing literature. 
Topics for future work include investigating an alternative path using the Interchangeability Principle \cite{shapiro2009lectures} and extensions to partial state information and universally measurable policies, e.g., by building on techniques from \cite{yu2021average}. 
